\documentclass{article}
\usepackage[margin=3cm]{geometry}
\usepackage[english]{babel}
\usepackage[T1]{fontenc}
\usepackage[utf8]{inputenc}
\usepackage{amssymb}
\usepackage{amsmath}
\usepackage{amsthm}
\usepackage{lmodern}
\usepackage[
  pdftex,
  bookmarks=true,%                   %%% generate bookmarks ...
  bookmarksnumbered=true,%           %%% ... with numbers
  hypertexnames=false,%              %%% needed for correct links to figures !!!
  breaklinks=true,%                  %%% breaks lines, but links are very small
  linkbordercolor={0 0 1}%          %%% blue frames around links]{hyperref}
  ]{hyperref}
\usepackage{algorithm}
\usepackage{algpseudocode}
\usepackage{color,todonotes}
\usepackage{mathdots}

\newcommand{\R}{\mathbb{R}}

\newcommand{\N}{\mathbb{N}}

\newcommand{\Q}{\mathbb{Q}}

\newcommand{\K}{\mathbb{K}}

\newcommand{\infnorm}[1]{\left\lVert{#1}\right\rVert}
\newcommand{\sigmap}[1]{{\Sigma{#1}}}

\newcommand{\IntI}{\operatorname{Int}}

\newcommand{\poly}{\operatorname{poly}}

\numberwithin{figure}{section}
\newtheorem{definition}{Definition}

\newtheorem{theorem}[definition]{Theorem}
\newtheorem{proposition}[definition]{Proposition}

\newtheorem{lemma}[definition]{Lemma}

\makeatletter
\let\c@algorithm\c@definition
\makeatother

\expandafter\let\expandafter\oldproof\csname\string\proof\endcsname
\let\oldendproof\endproof
\renewenvironment{proof}[1][\proofname]{%
  \oldproof[\bfseries #1]%
}{\oldendproof}

\title{Computational complexity of solving polynomial
differential equations over unbounded domains with non-rational coefficients}

\author{Amaury Pouly}

\begin{document}

\maketitle

\begin{abstract}
In this note, we extend the result of \cite{PoulyG16} about the complexity of solving polynomial
differential equations over unbounded domains to work with non-rational input.
In order to deal with arbitrary input, we phrase the result in framework of Conputable Analysis
\cite{Ko91}. As a side result, we also get a uniform result about complexity of the
operator, and not just about the solution.
\end{abstract}

The complexity of solving this kind of differential equation has been heavily studied over compact
domains but there are few results over unbounded domains. In \cite{PoulyG16} we
studied the complexity of this problem over unbounded domains and obtained a
bound that involved the length of the solution curve. Unfortunately, the result
was written for rational inputs only. In this note, we extend
it to work with any numbers, in the framework of Computable Analysis. To do so,
we will need to recall a few lemmas and introduce some notation. For any continous
function $y$, define
\[\IntI_y(a,b,\varepsilon)=\int_a^bk\sigmap{p}\max(1,\varepsilon+\infnorm{y(u)})^{k-1}du\]
and
\[\ell_y(a,b)=\int_a^b\sigmap{p}\max(1,\infnorm{y(u)})^kdu.\]
For any multivariate polynomial $p(x)=\sum_{|\alpha|\leqslant k} a_\alpha x^\alpha$,
we call $k$ the degree and denote the sum of the norm of the coefficients by
$\sigmap{p}=\sum_{|\alpha|\leqslant k}\infnorm{a_\alpha}$. Note that a vector
of polynomials can be identified to a vector with vector coefficients (i.e.
$\K^d[\R^n]$ is isomorphic to $(\K[\R^n])^d$) and always make this transformation
implicitly below. For such a polynomial $p$ and $\eta\geqslant0$, we call a
$\eta$-\emph{relative-approximation} of $p$ any polynomial
$\tilde{p}=\sum_{|\alpha|\leqslant k} \tilde{a}_\alpha x^\alpha$ with the same degree
such that $\infnorm{\tilde{a}_\alpha-a_\alpha}\leqslant\eta\infnorm{a_\alpha}$
for all $|\alpha|\leqslant k$. It follows almost by definition that:

\begin{lemma}\label{lem:eta_approx_poly}
If $\tilde{p}$ is a $\eta$-\emph{relative-approximation} of $p\in\R^n[\R^d]$ then
for all $x\in\R^d$ we have $\infnorm{\tilde{p}(x)-p(x)}\leqslant\eta\sigmap{p}\max(1,\infnorm{x})^k$
where $k$ is the degree of $p$.
\end{lemma}

We also recall the following simple lemma about polynomials.

\begin{lemma}[\cite{PoulyG16}]\label{lem:lip_poly}
Let $p\in\R^n[\R^d]$ and $k$ its degree. For all $a,b\in\R^d$ we have
\[\infnorm{p(b)-p(a)}\leqslant k\sigmap{p}\infnorm{b-a}\max(\infnorm{a},\infnorm{b})^{k-1}.\]
\end{lemma}

We will need to quantity to divergence between two PIVPs with slightly
different initial conditions and errors in the coefficients of the polynomials.

\begin{proposition}\label{prop:pivp_divergence}
Let $I=[a,b]$ be an interval, $p\in\R^n[\R^{n}]$ and $k$ its degree,
$y_0,\tilde{y}_0\in\R^n$ and $\tilde{p}$ a $\eta$-relative-approximation of $p$ for some $\eta\geqslant0$.
Assume that $y,\tilde{y}:I\rightarrow\R^n$ satisfies for all $t\in I$
\[\left\{\begin{array}{@{}r@{}l}y(0)&=y_0\\y'(t)&=p(y(t))\end{array}\right.\qquad
\left\{\begin{array}{@{}r@{}l}\tilde{y}(0)&=\tilde{y}_0\\\tilde{y}'(t)&=\tilde{p}(\tilde{y}(t))\end{array}\right..\]
For any $\varepsilon>0$ and $t\in I$, let
\[
\mu_\varepsilon(t)=\big(\infnorm{\tilde{y}_0-y_0}+\eta\ell_y(a,t)\big)
\exp\left((1+\eta)\IntI_y(a,t,\varepsilon)\right).
\]
If $\mu_\varepsilon(t)<\varepsilon$ then $\infnorm{z(t)-y(t)}\leqslant\mu_\varepsilon(t)$.
Furthermore, if the existence of $\tilde{y}$ is not known, then $\mu_\varepsilon(t)<\varepsilon$
implies that $\tilde{y}$ exists over $[a,b]$.
\end{proposition}

\begin{proof}
Let $\psi(t)=\infnorm{\tilde{y}(t)-y(t)}$. For any $t\in I$, we have
\[\psi(t)\leqslant\psi(a)+\int_a^t\infnorm{\tilde{p}(\tilde{y}(u))-p(y(u))}du.\]
Note that $\sigmap{\tilde{p}}\leqslant(1+\eta)\sigmap{p}$ and apply Lemmas~\ref{lem:eta_approx_poly}
and~\ref{lem:lip_poly} to get, for
$N(u)=\infnorm{y(u)}+\psi(u)$, that
\[\infnorm{\tilde{p}(\tilde{y}(u))-p(y(u))}\leqslant\eta\sigmap{p}\max(1,\infnorm{y(u)})^k+ k(1+\eta)\sigmap{p}N^{k-1}(u)\psi(u).\]
Putting everything together, we have
\[\psi(t)\leqslant\psi(a)+\int_a^t\eta\sigmap{p}\max(1,\infnorm{y(u)})^kdu
+\int_a^t(1+\eta)k\sigmap{p}N^{k-1}(u)\psi(u)du.\]
Apply the Generalized Gronwall's Inequality, using that the integral of non-negative
values is non-decreasing, to get
\[\psi(t)\leqslant
\left(\infnorm{\tilde{y}_0-y_0}+\int_a^t\eta\sigmap{p}\max(1,\infnorm{y(u)})^kdu\right)
\exp\left(\int_a^t(1+\eta)k\sigmap{p}N^{k-1}(u)du\right).\]

Define $t_1=\max\big\{t\in I\thinspace|\thinspace \forall u\in[a,t], \psi(u)\leqslant\varepsilon\big\}$
which is well-defined as the maximum of a closed and non-empty set ($a$ belongs to it).
Then for all $t\in[0,t_1]$, $N(t)\leqslant\infnorm{y(t)}+\varepsilon$ and thus:
\begin{align*}
\psi(t)
    &\leqslant\left(\infnorm{\tilde{y}_0-y_0}+\int_a^t\eta\sigmap{p}\max(1,\infnorm{y(u)})^{k-1}du\right)
        \exp\left(\int_a^t(1+\eta)k\sigmap{p}(\infnorm{y(u)}+\varepsilon)^{k-1} du\right)\\
    &\leqslant\big(\infnorm{\tilde{y}_0-y_0}+\eta\ell_y(0,t)\big)
        \exp\left((1+\eta)\IntI_y(0,t,\varepsilon)\right)\\
    &\leqslant\mu_\varepsilon(t).
\end{align*}
We will show by contradiction that $t_1=b$, which proves the result. Assume by
contradiction that $t_1<b$. Then by continuity of $\psi$ and because $\psi(a)=\mu(a)<\varepsilon$,
there exists $t_0\leqslant t_1$ such that $\psi(t_0)=\varepsilon$. But
then $t_0\in[0,t_1]$ so $\psi(t_0)\leqslant\mu(t)<\varepsilon$ by hypothesis, which is impossible.

To show the existence, assume by contradiction $\tilde{y}$ does not exists over $[a,b]$.
Apply Cauchy-Lipschitz theorem to get a maximal solution $\tilde{y}$ that exists over $[a,c[$
but not $[a,c]$ where $c\in[a,b]$. It is a well-known fact that $\infnorm{\tilde{y}(t)}\rightarrow+\infty$
as $t\rightarrow c$. Since $[a,b]$ is compact, $y$ is bounded over $[a,b]$.
It follows that $\infnorm{\tilde{y}(t)-y(t)}\rightarrow+\infty$ as $t\rightarrow c$.
Thus by continuity, there exists $d\in[a,c[$ such that $\infnorm{\tilde{y}(d)-y(d)}=\varepsilon$.
But then $\tilde{y}$ exists over $[a,d]$ so we can apply the above reasoning over $[a,d]$
to get that $\infnorm{\tilde{y}(d)-y(d)}\leqslant\mu_\varepsilon(d)$ since
$\mu_\varepsilon(d)\leqslant\mu_\varepsilon(b)<\varepsilon$.
It follows that $\infnorm{\tilde{y}(d)-y(d)}<\varepsilon$ which is impossible.
\end{proof}

We will need a result on the growth of the PIVP that only involves the initial condition.

\begin{proposition}\label{prop:growth_pivp}
Let $I=[a,b]$ be an interval, $p\in\R^n[\R^{n}]$ and $k$ its degree and $y_0\in\R^n$.
Assume that $y:I\rightarrow\R^n$ satisfies for all $t\in I$ that
\[y(a)=y_0\qquad y'(t)=p(y(t)),\]
then
\[\infnorm{y(t)-y(a)}\leqslant\frac{\alpha M|t-a|}{1-M|t-a|}\]
for every $t$ such that $M|t-a|<1$
where $M=(k-1)\sigmap{p}\alpha^{k-1}$ and $\alpha=\max(1,\infnorm{y_0})$.
\end{proposition}

\begin{proof}This is a consequence of Theorem~5 (Taylor approximation for PIVP) in \cite{PoulyG16},
restating an original result in \cite{WWSPC06}.
\end{proof}

We now recall the complexity result in \cite{PoulyG16}. For reasons that will
appear later, we will use the algorithm with ``hint'' rather than the full algorithm.

\begin{theorem}[Solving PIVPs with hint, \cite{PoulyG16}]\label{th:solving_pivp_hint}
There exists an algorithm $\mathcal{A}$ such that the following holds.
Let $a,b\in\Q$, $p\in\Q^n[\R^{n}]$ and $k$ its degree and $y_0\in\Q^n$.
Assume that $y:[a,b]\rightarrow\R^n$ satisfies for all $t\in[a,b]$ that
\[y(a)=y_0\qquad y'(t)=p(y(t)).\]
Let $I,\varepsilon\in\Q$ and $x=\mathcal{A}(a,y_0,p,b,\varepsilon,I)$, then
\begin{itemize}
\item either $x=\bot$ or $\infnorm{y(b)-x}\leqslant\varepsilon$,
\item if $I\geqslant 6\IntI_y(a,b,\varepsilon)$ then $x\neq\bot$,
\item if $I<\IntI_y(a,b,\varepsilon)$ then $x=\bot$,
\item the algorithm computes $x$ in time bounded in by
\[\poly\big(k,I,\log\ell_y(a,b),\log\infnorm{y_0},\log\sigmap{p},-\log\varepsilon\big)^n.\]
\end{itemize}
\end{theorem}

\begin{proof}
This is a consequence of various results in \cite{PoulyG16}.
The first two points follows from Lemma~10 (Algorithm is correct) and the third
one follows from the proof of Lemma~10 (but is not stated in the Lemma itself).
The fourth point is a consequence of Lemma 14~(Complexity of SolvePIVPVariable).
\end{proof}

For technical reasons, the previous lemma is not entirely satisfactory because the hint
$I$ is related to $\IntI_y$ but we would prefer that it relates to $\ell_y$. This
is possible thanks to a small trick.

\begin{lemma}\label{lem:solving_pivp_hint_len}
There exists an algorithm $\mathcal{B}$ such that the following holds.
Let $a,b\in\Q$, $p\in\Q^n[\R^{n}]$ and $k$ its degree and $y_0\in\Q^n$.
Assume that $y:[a,b]\rightarrow\R^n$ satisfies for all $t\in[a,b]$ that
\[y(0)=y_0\qquad y'(t)=p(y(t)).\]
Let $L,\varepsilon\in\Q$ and $x=\mathcal{B}(a,y_0,p,b,\varepsilon,L)$, then
\begin{itemize}
\item either $x=\bot$ or $\infnorm{y(b)-x}\leqslant\varepsilon$,
\item if $L\geqslant 12(k+1)\ell_y(a,b)$ then $x\neq\bot$,
\item if $L<\ell_y(a,b)$ then $x=\bot$,
\item the algorithm computes $x$ in time bounded in by
\[\poly\big(k,L,\log\ell_y(a,b),\log\infnorm{y_0},\log\sigmap{p},-\log\varepsilon\big)^n.\]
\end{itemize}
Furthermore, even if there no solution $y$ to the system over $[a,b]$, the algorithm 
always returns $\bot$ in time bounded by
\[\poly\big(k,L,\log\infnorm{y_0},\log\sigmap{p},-\log\varepsilon\big)^n.\]
\end{lemma}

\begin{proof}
Let $\mathcal{A}$ be the algorithm from Theorem~\ref{th:solving_pivp_hint}. The
hint of $\mathcal{A}$ is related to $\IntI_y$ which contains the integral of $\max(1,\infnorm{y(t)})^{k-1}$.
On the other hand, we would like to related to $\ell_y$ which contains the integral
of $\max(1,\infnorm{y(t)})^{k}$. So if we could increase the degree artifically by one,
without changing the complexity too much, we would almost have what we want.
The idea is to add one component that will always be $0$ but with a polynomial
of degree $k+1$. One possibility is $z'=z^{k+1}$ with $z(0)=0$ but it will be more
convenient to take $z'=\sigmap{p}z^{k+1}$.

Without loss of generality, we assume that $\varepsilon\leqslant\tfrac{1}{4k}$.
Given the hypothesis of the lemma, let
\[z_0=(y_0,0),\qquad q(y,z)=(p(y),\sigmap{p}z^{k+1}).\]
and define
\[\mathcal{B}(a,y_0,p,b,\varepsilon,L)=\mathcal{A}(a,z_0,q,b,\varepsilon,L)_{1..n}.\]
It is clear from the definition that the only solution of
\[z(0)=z_0\qquad z'(t)=q(z(t))\]
is of the form $z(t)=(y(t),0)$.
We will now check that $\mathcal{B}$ satisfies the claim. Let $x=\mathcal{A}(a,y_0,p,b,\varepsilon,L)$.
First, recall that $\sigmap{q}$ is the maximum of all components of $q$, and since
$\sigmap{(z\mapsto\sigmap{p}z^{k+1})}=\sigmap{p}$ we get that $\sigmap{q}=\sigmap{p}$.
Furthermore, $q$ is of degree $k+1$ and $\infnorm{z(t)}=\infnorm{y(t)}$ for all $t\in[a,b]$.
\begin{itemize}
\item By definition of $\mathcal{A}$, either $x=\bot$ (and thus $x_{1..n}=\bot$)
    or $\infnorm{x-z(t)}\leqslant\varepsilon$, but since $z(t)=(y(t),0)$ then $\infnorm{x_{1..n}-y(t)}\leqslant\varepsilon$.
\item If $L\geqslant 12(k+1)\ell_y(a,b)$ then
    \begin{align*}
    6\IntI_z(a,b,\varepsilon)
        &=6\int_a^b(k+1)\sigmap{q}\max(1,\varepsilon+\infnorm{z(u)})^{(k+1)-1}du\\
        &=6(k+1)\int_a^b\sigmap{p}\max(1,\varepsilon+\infnorm{y(u)})^kdu\\
        &\leqslant6(k+1)(1+\varepsilon)^k\int_a^b\sigmap{p}\max(1,\infnorm{y(u)})^kdu\\
        &\leqslant6(k+1)(1+\tfrac{1}{4k})^k\ell_y(a,b)\\
        &\leqslant12(k+1)\ell_y(a,b)\\
        &\leqslant L.
    \end{align*}
    Thus $x\neq\bot$ by Theorem~\ref{th:solving_pivp_hint}.
\item If $L<\ell_y(a,b)$ then
    \begin{align*}
    L
        &<\int_a^b\sigmap{p}\max(1,\infnorm{y(u)})^kdu\\
        &=\int_a^b\sigmap{q}\max(1,\infnorm{z(u)})^kdu\\
        &\leqslant\int_a^b(k+1)\sigmap{q}\max(1,\varepsilon+\infnorm{z(u)})^{(k+1)-1}du\\
        &=\IntI_z(0,t,\varepsilon).
    \end{align*}
    Thus $x=\bot$ by Theorem~\ref{th:solving_pivp_hint}.
\item By Theorem~\ref{th:solving_pivp_hint}, the complexity is bounded by
    \[\poly\big(k+1,L,\log\ell_z(a,b),\log\infnorm{z_0},\log\sigmap{q},-\log\varepsilon\big)^{n+1}.\]
    Recall that for any $t\in[a,b]$ we have
    \begin{align*}
    \infnorm{y(t)}
        &\leqslant\infnorm{y_0}+\int_0^t\infnorm{y'(u)}du\\
        &=\infnorm{y_0}+\int_0^t\infnorm{p(y(u))}du\\
        &\leqslant\infnorm{y_0}+\int_0^t\sigmap{p}\max(1,\infnorm{y(u)})^kdu\\
        &=\infnorm{y_0}+\ell_y(0,t)\\
        \leqslant\infnorm{y_0}+\ell_y(0,b).
    \end{align*}
    Thus
    \begin{align*}
    \ell_z(a,b)
        &=\int_a^b\sigmap{q}\max(1,\infnorm{z(u)})^{k+1}du\\
        &=\int_a^b\sigmap{p}\max(1,\infnorm{y(u)})^{k+1}du\\
        &\leqslant\max(1,\infnorm{y_0}+\ell_y(a,b))\int_a^b\sigmap{p}\max(1,\infnorm{y(u)})^kdu\\
        &\leqslant\max(1,\infnorm{y_0}+\ell_y(a,b))\ell_y(a,b).
    \end{align*}
    It follows that the complexity is bounded by
    \[\poly\big(k,L,\log\ell_y(a,b),\log\infnorm{y_0},\log\sigmap{p},-\log\varepsilon\big)^n.\]
\end{itemize}

The extra statement is a consequence of two facts. First, disregarding the existence
or not of $y$, if $b'<b$ and
$\mathcal{A}(a,z_0,q,b',\varepsilon,L)=\bot$ then $\mathcal{A}(a,z_0,q,b,\varepsilon,L)=\bot$.
This is a consequence of the fact that the algorithm does not use $b$ in any intermediate
computation except to check if it has reached time $b$. In other words, the algorithm
will perform exactly the same on the two instances and thus return $\bot$ in both.
We refer the reader to Algorithm~11 in \cite{PoulyG16} to check the details of this claim.
Furthermore, it follows from this that the running of the algorithm on both
instances is the same (they execute exactly the same number of instructions).

Second, by the Cauchy-Lipschitz theorem, there exists a maximal solution $y$ whose domain
is open and contains a neighbourhood of $a$. Thus there exists a $c\in]a,b]$ such
that $y$ is defined over $[a,c[$ but not in $c$. It is a well-known fact that
$\infnorm{y(t)}\rightarrow+\infty$ as $t\rightarrow c$. Since,
as we saw above, $\infnorm{y(t)}\leqslant\infnorm{y_0}+\ell_y(0,t)$,
it follows that $\ell_y(a,t)\rightarrow+\infty$ as $t\rightarrow c$. Thus by continuity,
there exists $b'\in[a,c[$ such that $\ell_y(a,b')=L+1$. But then, by
the third point above (and since $y$ exists over $[a,b']$),
\[\mathcal{A}(a,z_0,q,b',\varepsilon,L)=\bot.\]
And since $b'<b$, it follows that $\mathcal{A}(a,z_0,q,b,\varepsilon,L)=\bot$
by the claim above. Furthermore, since we saw earlier that the complexity of both
instances is the same, it follows that it returns $\bot$ in time bounded by
\[\poly\big(k,L,\log\ell_y(a,b'),\log\infnorm{y_0},\log\sigmap{p},-\log\varepsilon\big)^n.\]
which satisfies the claim since $\ell_y(a,b')=L+1$.
\end{proof}

We are now ready to state and prove a result about the complexity of solving PIVPs
for any inputs.

\begin{theorem}[Complexity of Solving PIVPs]\label{th:pivp_comp_analysis}
Let $I=[a,b]$ be an interval, $p\in\R^d[\R^{d}]$ and $k$ its degree and $y_0\in\R^d$.
Assume that $y:I\rightarrow\R^d$ satisfies for all $t\in I$ that
\begin{equation}\label{eq:ode}
y(a)=y_0\qquad y'(t)=p(y(t)),
\end{equation}
then $y(b)$ can be computed with precision $2^{-\mu}$ in time bounded by
\begin{equation}\label{eq:pivp_comp_analysis_bound}
\poly(k,\ell_y(a,b),\log\infnorm{y_0},\log\sigmap{p},\mu)^d.
\end{equation}
More precisely, there exists a Turing machine $\mathcal{M}$ such that for any oracle
$\mathcal{O}$ representing\footnote{See \cite{Ko91} for more details. In short, the machine can ask arbitrary approximation
of $a, y_0, p$ and $b$ to the oracle. The polynomial is represented by the finite list of coefficients.} $(a,y_0,p,b)$ and any $\mu\in\N$,
$\infnorm{\mathcal{M}^\mathcal{O}(\mu)-y(b)}\leqslant2^{-\mu}$
where $y$ satisfies \eqref{eq:ode}, and the number of steps of the machine is bounded by \eqref{eq:pivp_comp_analysis_bound}
for all such oracles.
\end{theorem}

\begin{proof}
Let $\mathcal{B}$ be the algorithm from Lemma~\ref{lem:solving_pivp_hint_len}.
Without loss of generality we assume that $a\in\Q$ (since we can always replace $a$ by $0$
and $b$ by $b-a$).
Let $\mathcal{O}$ be an oracle for $a, y_0, p$ and $b$ (where $p$ is represented
by the finite list of its coefficients) and $\mu$ the input of
the machine. Let $\varepsilon\in\Q$ such that $\varepsilon<e^{-\mu-\ln3}$.
Define, for all $n\in\N$:
\begin{itemize}
\item $L_n=n$,
\item $\nu_n=e^{-4kL_n-\ln2}\varepsilon$,
\item $y_0^{(n)}\in\Q^n$ be such that $\infnorm{y_0^{(n)}-y_0}\leqslant\nu_n$,
\item $\eta_n\in\Q^n$ be such that $\eta_n\leqslant\tfrac{\nu_n}{L_n}$ and $\eta_n<1$,
\item $p^{(n)}$ be a $\eta_n$-relative-approximation of $p$,
\item $t^{(n)}\in\Q$ be such that $t^{(n)}\leqslant b$ and
\[b-t^{(n)}\leqslant\frac{\varepsilon}{2k\sigmap{p}\max(1,\infnorm{y_0}+L_n)^k}.\]
\end{itemize}
Finally define the sequence
\[x_n=\mathcal{B}(a,y_0^{(n)},p^{(n)},t^{(n)},\varepsilon,L_n)\]
and let $y^{(n)}$ be the maximal solution of
\[y^{(n)}(a)=y_0^{(n)}\qquad {y^{(n)}}'=p^{(n)}(y^{(n)}).\]
Note that by the Cauchy-Lipschitz theorem, we know such a solution exists but
it may not exists over $[a,b]$. Note, and this is a consequence of Lemma~\ref{lem:solving_pivp_hint_len},
that we can safely apply $\mathcal{B}$ to a system even if we don't know that its
solution exists over $[a,b]$.

First, we claim that if $L_n\geqslant\ell_{y}(a,b)$ then $y^{(n)}$ exists
over $[a,t^{(n)}]$ and $\infnorm{y(u)-y^{(n)}(u)}\leqslant\varepsilon$ for
all $u\in[a,t^{(n)}]$. Indeed,
assume that $L_n\geqslant\ell_{y}(a,b)$. Then
\[L_n\geqslant\ell_{y}(a,b)\geqslant\ell_{y}(a,t^{(n)}).\]
Let
\[\mu_\varepsilon(t)=\left(\infnorm{y_0^{(n)}-y_0}+\eta_n\ell_y(a,t)\right)
\exp\left((1+\eta_n)\IntI_y(a,t,\varepsilon)\right).\]
Apply Lemma~13 (Relationship between Int and Len) in \cite{PoulyG16} to get that
\[\IntI_y(a,t,\varepsilon)\leqslant2k\ell_y(a,t).\]
It follows that
\begin{align*}
    \mu_\varepsilon(t^{(n)})
        &\leqslant\left(\infnorm{y_0^{(n)}-y_0}+\eta_n\ell_y(a,t^{(n)})\right)
            \exp\left((1+\eta_n)2k\ell_y(a,t^{(n)})\right)\\
        &\leqslant\left(\nu_n+\eta_nL_n\right)\exp\left((1+\eta_n)2kL_n\right)\\
        &\leqslant2\nu_n\exp\left(4kL_n\right)\\
        &<\varepsilon.
\end{align*}
Apply Proposition~\ref{prop:pivp_divergence} to get that $y^{(n)}$ exists
over $[a,t^{(n)}]$. For all $u\in[a,t^{(n)}]$, note that $\mu_\varepsilon(u)\leqslant\mu_\varepsilon(t^{(n)})<\varepsilon$
and apply Proposition~\ref{prop:pivp_divergence} again over $[a,u]$ to get that
\[\infnorm{y(u)-y^{(n)}(u)}\leqslant\varepsilon.\]

Second, we claim that if $x_n\neq\bot$ then $\infnorm{x_n-y(b)}\leqslant e^{-\mu}$.
Indeed, by Lemma~\ref{lem:solving_pivp_hint_len}, if $x_n\neq\bot$ then it must be the
case that
\[L_n\geqslant\ell_{y}(a,b).\]
Apply the first claim to get that $y^{(n)}$ exists over $[a,t^{(n)}]$ and that
\[\infnorm{y(t^{(n)})-y^{(n)}(t^{(n)})}\leqslant\varepsilon.\]
Apply Lemma~\ref{lem:solving_pivp_hint_len} to get that
\[\infnorm{x_n-y^{(n)}(t^{(n)})}\leqslant\varepsilon.\]
It remains to see the relationship between $y(b)$ and $y(t^{(n)})$.
Recall that
\[\infnorm{y(t^{(n)})}\leqslant\infnorm{y_0}+\ell_y(a,t^{(n)})\leqslant\infnorm{y_0}+L_n.\]
Let $M=(k-1)\sigmap{p}\alpha^{k-1}$ and $\alpha=\max(1,\infnorm{y(t^{(n)})})$.
Note that $\alpha\leqslant\max(1,\infnorm{y_0}+L_n)$.
It follows by definition of $t^{(n)}$ that
\begin{align*}
    M|t-t^{(n)}|
        &=(k-1)\sigmap{p}\alpha^{k-1}|t-t^{(n)}|\\
        &\leqslant k\sigmap{p}\max(1,\infnorm{y_0}+L_n)^{k-1}|t-t^{(n)}|\\
        &\leqslant\frac{\varepsilon}{2\max(1,\infnorm{y_0}+L_n)}\\
        &\leqslant\frac{1}{2}<1.
\end{align*}
Thus we can apply Proposition~\ref{prop:growth_pivp} to $y$ with $a=t^{(n)}$ to get that
\[\infnorm{y(b)-y(t^{(n)})}\leqslant\frac{\alpha M|b-t^{(n)}|}{1-M|b-t^{(n)}|}.\]
Consequently
\begin{align*}
\infnorm{y(b)-y(t^{(n)})}
    &\leqslant\frac{\alpha M|t-t^{(n)}|}{1-M|t-t^{(n)}|}\\
    &\leqslant\frac{\alpha \frac{\varepsilon}{2\max(1,\infnorm{y_0}+L_n)}}{1-1/2}\\
    &\leqslant\varepsilon.
\end{align*}
Putting everything together, we get that
\[\infnorm{x_n-y(b)}\leqslant3\varepsilon\leqslant e^{-\mu}.\]

Third, we claim that if $L_n\geqslant 48(k+1)\ell_y(a,b)$ then $x_n\neq\bot$. Indeed, assume
that this is the case. Then in particular $L_n\geqslant\ell_y(a,b)$ so by the first fact,
$y^{(n)}$ exists over $[a,t^{(n)}]$ and for all $t\in[a,t^{(n)}]$ we have
\[\infnorm{y(t)-y^{(n)}(t)}\leqslant\varepsilon.\]
It follows from this that
\begin{align*}
\ell_{y^{(n)}}(a,t^{(n)})
    &=\int_a^{t^{(n)}}\sigmap{p^{(n)}}\max\left(1,\infnorm{y^{(n)}(u)}\right)^kdu\\
    &\leqslant\int_a^{t^{(n)}}(1+\eta_n)\sigmap{p}\max(1,\infnorm{y(u)}+\varepsilon)^kdu\\
    &\leqslant(1+\eta_n)(1+\varepsilon)^k\int_a^{t^{(n)}}\sigmap{p}\max(1,\infnorm{y(u)})^kdu\\
    &\leqslant2(1+\tfrac{1}{4k})^k\ell_y(a,t^{(n)})\\
    &\leqslant4\ell_y(a,b).
\end{align*}
Thus
\[L_n\geqslant48k\ell_y(a,b)\geqslant 12(k+1)\ell_{y^{(n)}}(a,t^{(n)})\]
and by Lemma~\ref{lem:solving_pivp_hint_len}, $x_n\neq\bot$.

Now consider the algorithm that computes the sequence $(x_n)_n$ and returns
the first $x_n\neq\bot$. Thanks to the second claim, this algorithm is correct
because if $x_n\neq\bot$ then $\infnorm{x_n-y(b)}\leqslant\varepsilon$. Furthermore
this algorithm terminates. Indeed, let $N$ be the smallest integer such that
\[L_N\geqslant 48(k+1)\ell_y(a,b).\]
It exists because $L_n\rightarrow+\infty$ as $n\rightarrow+\infty$.
Then $x_N\neq\bot$ and thus the algorithm terminates. Finally, we claim this algorithm
has the right complexity. Indeed, let $n_0$ be the first $n$ such that $x_{n_0}\neq\bot$.
By construction, $n_0\leqslant N$ and the algorithm computes $x_1,x_2,\ldots,x_{n_0}$
and returns. By Lemma~\ref{lem:solving_pivp_hint_len},
the complexity of computing $x_n$ for $n<n_0$ is bounded by
\[\poly\big(k,L_n,\log\infnorm{y_0},\log\sigmap{p},-\log\varepsilon\big)^d\]
since $x_n=\bot$. Furthermore, the complexity of computing $x_{n_0}$ is bounded
by
\[\poly\big(k,L_{n_0},\log\ell_y(a,b),\log\infnorm{y_0},\log\sigmap{p},-\log\varepsilon\big)^d.\]
Since $n_0\leqslant N$, it follows that $L_n\leqslant L_N$ for all $n\leqslant n_0$ and
thus the total complexity is bounded by
\[\sum_{n=1}^{n_0}\poly\big(k,L_N,\log\ell_y(a,b),\log\infnorm{y_0},\log\sigmap{p},-\log\varepsilon\big)^d.\]
Furthermore, since $L_n=n$ and $N$ is the smallest integer such that $L_N\geqslant 48(k+1)\ell_y(a,b)$,
it must be the case that
\[L_N<49(k+1)\ell_y(a,b)\]
and thus that
\[n_0\leqslant N<49(k+1)\ell_y(a,b).\]
Putting everything together, we get that the total complexity is bounded by
\begin{align*}
    \sum_{n=1}^{n_0}&\poly\big(k,96(k+1)\ell_y(a,b),\log\ell_y(a,b),\log\infnorm{y_0},\log\sigmap{p},-\log\varepsilon\big)^d\\
        &\leqslant\sum_{n=1}^{n_0}\poly\big(k,\ell_y(a,b),\log\infnorm{y_0},\log\sigmap{p},\mu+\ln3\big)^d\\
        &\leqslant n_0\poly\big(k,\ell_y(a,b),\log\infnorm{y_0},\log\sigmap{p},\mu\big)^d\\
        &\leqslant 49(k+1)\ell_y(a,b)\poly\big(k,\ell_y(a,b),\log\infnorm{y_0},\log\sigmap{p},\mu\big)^d\\
        &\leqslant \poly\big(k,\ell_y(a,b),\log\infnorm{y_0},\log\sigmap{p},\mu\big)^d.
\end{align*}
\end{proof}

Finally, we would like to remind the reader that the existence of a solution $y$
of a PIVP up to a given time is undecidable, see \cite{GBC07} more details. This explains
why, in the previous theore, we have so assume the existence of the solution if
we want to have any hope of computing it.

\bibliographystyle{alpha}
\bibliography{extracted}

\end{document}